\documentclass{article}
\usepackage[a4paper,margin=1in]{geometry}
\usepackage[T1]{fontenc}
\usepackage{setspace}
\usepackage{graphicx} % Required for inserting images
\usepackage{amsmath}
\usepackage{mathtools}
\usepackage{amsthm}
\usepackage{tablefootnote}
\usepackage{braket}
\usepackage{amssymb}
\usepackage{mathrsfs}
\usepackage{amsfonts}
\usepackage{authblk}
\usepackage{makecell}
\usepackage{threeparttable}
\usepackage{bbding}
\usepackage{subfig}
\usepackage{dcolumn}
\usepackage{float}
\usepackage{comment}
\usepackage{multirow}
\usepackage{color}
\usepackage{tcolorbox}
\usepackage{algorithm}
\usepackage{algorithmic}
\usepackage{tikz}
\usepackage{appendix}
\usepackage[linktocpage=true, colorlinks,
linkcolor=blue,citecolor=blue,
bookmarks,bookmarksopen,bookmarksnumbered]
{hyperref}
\usepackage{circuitikz} 
\usepackage{tikz-3dplot} 
\usepackage{adjustbox}
\allowdisplaybreaks[4]

\newcommand{\rbra}[1]{\left( #1 \right)} 
\newcommand{\sbra}[1]{\left[ #1 \right]}
\newcommand{\cbra}[1]{\left\{ #1 \right\}}
\newcommand{\abs}[1]{\lvert #1 \rvert}
\newcommand{\Abs}[1]{\lVert #1 \rVert}

\theoremstyle{plain}

\newtheorem{theorem}{Theorem}[section]

\newtheorem{lemma}[theorem]{Lemma}

\theoremstyle{definition}
\newtheorem{definition}[theorem]{Definition}

\newtheorem{remark}[theorem]{Remark}

\title{Quantum Tsallis Entropy Estimation: Near-Optimal Bounds and High-Precision Improvements}

\begin{document}

\author{Chengshen Gao}
\author{Zongxing Xiong\thanks{hongbe123@outlook.com}}
\author{Lvzhou Li\thanks{lilvzh@mail.sysu.edu.cn}}
\affil{Institute of Quantum Computing and Software, School of Computer Science and Engineering, Sun Yat-sen University, Guangzhou 510006, China}
\date{}

\maketitle

\begin{abstract} 
Estimating statistical properties of probability distributions is a
fundamental problem in statistics and theoretical computer science. We develop quantum algorithms for estimating the $q$-Tsallis entropy $\mathrm{H}_q^{\mathrm{Tsa}}(P)
=\frac{\sum_{i=1}^n p_i^q-1}{1-q}$ when $0<q<1$, given purified quantum query access to a distribution $P$ on $[n]$. When $0<q\leq \frac{1}{2}$, we provide a quantum algorithm with query complexity $\widetilde{O}\rbra{\left(\frac{n}{\varepsilon}\right)^{\frac{1}{2q}}}$, which is optimal up to polylogarithmic factors. When $\frac{1}{2}<q<1$,   we introduce a two-stage simultaneous probability estimation procedure, yielding a quantum algorithm with query complexity $\widetilde{O}\left(\frac{n^{\frac12+\frac{1}{2q}}}{\varepsilon^{\frac{1}{2q}}}+\frac{n^{1-q}}{\varepsilon}\right)$, which improves the result given by \hyperlink{cite.chen2026quantum}{Chen et al. (ICALP 2026)} when $\varepsilon=O\rbra{n^{1-2q}}$.

\begin{comment}
Estimating statistical properties of probability distributions is a
fundamental problem in statistics and computer science. We develop quantum algorithms for estimating the $q$-Tsallis entropy $\mathrm{H}_q^{\mathrm{Tsa}}(P)
=\frac{\sum_{i=1}^n p_i^q-1}{1-q}$ when $0<q<1$, given purified quantum query access to a distribution $P$ on $[n]$. When $0<q\leq \frac{1}{2}$, we provide a quantum algorithm that adapts the framework of \hyperlink{cite.chen2026quantum}{Chen et al. (ICALP 2026)} with query complexity $\widetilde{O}\rbra{\left(\frac{n}{\varepsilon}\right)^{\frac{1}{2q}}}$, which is optimal up to polylogarithmic factors. For $\frac{1}{2}<q<1$, however, this framework incurs a higher query complexity in the high-precision regime. To reduce this cost, we introduce a two-stage simultaneous probability estimation procedure, yielding an algorithm that uses $\widetilde{O}\left(\frac{n^{\frac12+\frac{1}{2q}}}{\varepsilon^{\frac{1}{2q}}}+\frac{n^{1-q}}{\varepsilon}\right)$ queries and improves upon the query complexity of \hyperlink{cite.chen2026quantum}{Chen et al. (ICALP 2026)} when $\varepsilon< O\rbra{n^{1-2q}}$.
\end{comment}
\end{abstract}

\section{Introduction}
Estimating numerical properties of an unknown probability distribution
is a fundamental problem in statistics and theoretical computer science.
An important family of such properties is given by the power sums
$$
F_q(P)=\sum_{i=1}^n p_i^q,
$$
which determine both the Tsallis entropy~\cite{tsallis1988possible} and the R\'enyi entropies~\cite{renyi1961measures}. In this paper, we study quantum algorithms for estimating the $q$-Tsallis entropy
\begin{equation}\label{eq:entropy}
\mathrm{H}_q^{\mathrm{Tsa}}(P)
=\frac{F_q(P)-1}{1-q}
\end{equation}
 for $0<q<1$, given purified quantum
query access to a distribution $P=(p_1,\ldots,p_n)$ on $[n]$.

The regime $0<q<1$ presents a particular challenge: small probabilities
can contribute substantially to $F_q(P)$ even though they account for little probability mass. Indeed, the ratio $p_i^q/p_i=p_i^{q-1}$ increases as $p_i$ decreases. Moreover, the derivative $q x^{q-1}$ diverges as $x$ approaches zero, making the resulting power-sum estimate sensitive to errors in small probabilities. An effective algorithm must therefore control the estimation error across different probability scales.

Quantum algorithms for distribution testing and entropy estimation have been studied extensively \cite{chakraborty2010new,bravyi2011quantum,montanaro2015quantum,gilyen2020distributional,luo2024succinct,canonne_et_al:LIPIcs.TQC.2025.7,li2018quantum,wang2024quantum}. Most closely related to our work, \cite{chen2026quantum} introduced a multilevel framework for estimating functionals of discrete distributions. Their framework partitions the probabilities $p_i$ into intervals and applies quantum singular value transformation (QSVT)~\cite{gilyen2019quantum,martyn2021grand} separately within each interval. For $0<q<1$, it yields a Tsallis entropy estimator with query complexity $\widetilde{O}\left(\frac{n^{\frac{1}{q}-\frac{1}{2}}}{\varepsilon^{\frac{1}{q}}}\right)$, while the lower bounds are $\widetilde{\Omega}\rbra{ \left(\frac{n}{\varepsilon}\right)^{\frac{1}{2q}}}$ for $0<q\leq \frac{1}{2}$ and $\widetilde{\Omega}\rbra{
\left(\frac{n}{\varepsilon}\right)^{\frac{1}{2q}}
+\frac{n^{1-q}}{\varepsilon}}$ for $\frac{1}{2}<q<1$. A gap remains between this upper bound and the lower bounds, both in the dependence on the distribution size $n$ and in the dependence on the additive error $\varepsilon$.
This raises the following question:
\begin{center}
  \emph{Can we improve the quantum query complexity of Tsallis entropy
estimation for $0<q<1$?}  
\end{center}
We address this question using different estimators for
$0<q\leq\frac12$ and $\frac12<q<1$, obtaining optimal query
complexity up to polylogarithmic factors in the former regime
and improved bounds at high precision in the latter.

\subsection{Main results}

We use the purified quantum query access model. The input is a
unitary oracle $U_P$ satisfying
\begin{equation}\label{eq:oracle}
U_P\ket{0}_A\ket{0}_B
=\sum_{i=1}^n\sqrt{p_i}\ket{\phi_i}_A\ket{i}_B,
\end{equation}
where each $\ket{\phi_i}$ is an arbitrary normalized state.
Query complexity counts applications of $U_P$, $U_P^\dagger$,
and their controlled versions.
Throughout, we assume $n\geq2$ and $0<\varepsilon\leq1$. Our results address the regimes $0<q\leq\frac12$ and $\frac12<q<1$ separately.

\paragraph{Optimal estimation for $0<q\leq\frac12$.}
By adapting the multilevel framework of Chen et al.~\cite{chen2026quantum}, we obtain the following result.
\begin{theorem}[Tsallis entropy estimation for $0<q\leq1/2$]
\label{thm: low}
Let $P=(p_1,\ldots,p_n)$ be a probability distribution over $[n]$, and
suppose that purified quantum access $U_P$ to $P$ is available. For any
$0<q\leq1/2$ and $0<\varepsilon\leq1$, there exists a quantum algorithm
that outputs an estimate $\widehat{H}$ satisfying $\left|
\widehat{H}
-
\mathrm{H}_q^{\text{Tsa}}(P)
\right|
\leq\varepsilon$ with probability at least $2/3$, using
$$
\widetilde{O}\left(
\left(\frac{n}{\varepsilon}\right)^{\frac{1}{2q}}
\right)
$$
queries to $U_P$ and $U_P^\dagger$.
\end{theorem}

This bound matches the known lower bound up to polylogarithmic factors. Our construction uses uniform initialization and modifies
the local QSVT polynomials, while retaining the discriminators
and amplitude estimation routine of the existing framework.

\paragraph{Improved estimation for $\frac12<q<1$.}
For this regime, we combine a two-stage simultaneous probability estimation
procedure with quantum mean estimation~\cite{montanaro2015quantum,hamoudi2021quantum,kothari2023mean}.

\begin{theorem}[Tsallis entropy estimation for $\frac12<q<1$]
\label{thm:high}
Let $\frac12<q<1$ and $0<\varepsilon\leq1$.
Given purified quantum query access to a distribution $P$ on $[n]$,
there exists a quantum algorithm that outputs an estimate
$\widehat{H}$ satisfying $\left|\widehat{H}-\mathrm{H}_q^{\mathrm{Tsa}}(P)\right|
\leq\varepsilon$ with probability at least $2/3$, using
\[
\widetilde{O}\left(
\frac{n^{\frac12+\frac{1}{2q}}}
     {\varepsilon^{\frac{1}{2q}}}
+\frac{n^{1-q}}{\varepsilon}
\right)
\]
queries to $U_P$ and $U_P^\dagger$.
\end{theorem}

The first term accounts for constructing estimates of the
probabilities, and the second for quantum mean estimation.
Compared with the bound
$\widetilde{O}\left(\frac{n^{\frac{1}{q}-\frac{1}{2}}}{\varepsilon^{\frac{1}{q}}}\right)$
of Chen et al.~\cite{chen2026quantum}, our bound improves the
dependence on $\varepsilon$, at the cost of a larger exponent
of $n$ in the first term. Ignoring polylogarithmic factors,
it gives an improvement when $\varepsilon= O\rbra{n^{1-2q}}$.

Table~\ref{tab:bounds} summarizes our query bounds and compares them with existing results.

\begin{table*}[htbp]
\caption{Comparison of quantum query complexities for estimating the
$q$-Tsallis entropy to additive error $\varepsilon$ under purified quantum
access.}
\centering
\begin{tabular}{cccc}
\hline
\multicolumn{2}{c}{}
& Existing results 
& This work \\
\hline
& & & \\[-5pt]

\multirow{2}{*}{$0<q\leq1/2$}
& Upper bound
& $\widetilde{O}\left(\frac{n^{\frac{1}{q}-\frac{1}{2}}}{\varepsilon^{\frac{1}{q}}}\right)$\cite{chen2026quantum}
& $\widetilde{O}\left(
\left(\frac{n}{\varepsilon}\right)^{\frac{1}{2q}}
\right)$
(Theorem~\ref{thm: low})
\\[5pt]

& Lower bound
& $\widetilde{\Omega}\rbra{
\left(\frac{n}{\varepsilon}\right)^{\frac{1}{2q}}}$\cite{jiao2015minimax}
& -- \\
[8pt]

\hline
& & & \\[-5pt]

\multirow{2}{*}{$1/2<q<1$}
& Upper bound
& $\widetilde{O}\left(\frac{n^{\frac{1}{q}-\frac{1}{2}}}{\varepsilon^{\frac{1}{q}}}\right)$\cite{chen2026quantum}
& $\widetilde{O}\rbra{
\frac{n^{\frac{1}{2}+\frac{1}{2q}}}{\varepsilon^{\frac{1}{2q}}}
+\frac{n^{1-q}}{\varepsilon}}$
(Theorem~\ref{thm:high})
\\[5pt]

& Lower bound
& $\widetilde{\Omega}\rbra{
\left(\frac{n}{\varepsilon}\right)^{\frac{1}{2q}}
+\frac{n^{1-q}}{\varepsilon}}$\cite{jiao2015minimax}
& -- \\
[5pt]

\hline
\end{tabular}
\label{tab:bounds}
\end{table*}

\subsection{Techniques overview}

To outline the main technical ideas, we consider two parameter regimes separately. We first explain how the multilevel framework yields a query complexity that is optimal up to polylogarithmic factors for $0<q\leq\frac12$, and why it incurs a higher query cost for $\frac12<q<1$. We then present the key ingredients of our algorithm and analysis for $\frac12<q<1$, showing how they lead to an improved query complexity.

\subsubsection{Proof techniques for $0<q\leq\frac12$}

We adapt the multilevel framework of \cite{chen2026quantum} through uniform reweighting, which yields the optimal query bound for $0<q\leq\frac12$.
We then discuss why this approach incurs a higher query cost
when $\frac12<q<1$.
\paragraph{Uniform reweighting for $0<q\leq\frac12$.}
A natural starting point is that
\[
F_q(P)=\sum_{i=1}^n p_i\,p_i^{q-1}.
\]
With distribution-weighted initialization, a QSVT-based estimator
encodes a normalized approximation of $p_i^{q-1}$ in the amplitudes. Since the corresponding singular value is proportional
to $\sqrt{p_i}$, the required polynomial approximates a function
proportional to $x^{q-1}$.
For $0<q<1$, this function diverges at the origin, whereas a QSVT
polynomial must remain bounded on $[-1,1]$.
The multilevel framework of \cite{chen2026quantum} handles this
difficulty by applying separate transformations at different
probability scales. However, the normalization required at small
scales introduces rescaling factors that amplify the estimation
error.

Our observation is that changing the input weights changes the
function that QSVT must approximate. We replace distribution-weighted
initialization with the uniform state
\[
\ket{u}=\frac{1}{\sqrt n}\sum_{i=1}^n\ket{\ell_i}.
\]
Each label now receives weight $1/n$, so the local transformations
need only approximate a normalized version of $x^q$.
More precisely, within the $j$-th probability interval, we choose
a polynomial $Q_j$ and a rescaling factor $\Lambda_j$ such that 
$
\Lambda_j Q_j(\sqrt{p_i}/2)^2\approx p_i^q
$. Unlike $x^{q-1}$, the function $x^q$ is bounded and vanishes
at the origin. Uniform initialization also assigns greater weight
to labels with $p_i<1/n$, allowing their contribution to the
power sum to be resolved more efficiently.

Let $t_j$ denote the geometrically decreasing probability scales.
The resulting query cost is bounded by $\widetilde{O}\left(
\frac{n}{\varepsilon}\sum_{j=1}^m t_j^{q-1/2}
\right)$. We choose the finest scale as $t_m=\Theta\left(
\left(\frac{\varepsilon}{n}\right)^{1/q}
\right)$, so that probabilities below this scale contribute only
$O(\varepsilon)$ to the power sum.
For $q\leq \frac12$, the geometric sum is dominated by its finest
scale; for $q=\frac12$, it contributes only the number of levels.
This gives the query bound
\[
\widetilde{O}\left(
\left(\frac{n}{\varepsilon}\right)^{\frac{1}{2q}}
\right).
\]
\paragraph{The limitation of uniform reweighting for $q>1/2$.}
Uniform initialization also introduces a cost: converting
an expectation under the uniform distribution into the power
sum requires multiplication by $n$, which amplifies the
estimation error.
For $q>1/2$, the geometric sum above is no longer dominated
by the finest scale. A more refined analysis that accounts
for the acceptance probabilities gives $\widetilde{O}\!\left(
\frac{n^{3/2-q}}{\varepsilon}
\right)$. This bound is a factor of $\sqrt n$ larger than
$\frac{n^{1-q}}{\varepsilon}$. Thus, removing the singular QSVT
transformation does not by itself give the desired dependence
on the accuracy in this regime.
This motivates a different approach: first construct classical
estimates of the probabilities, and then use them in quantum
mean estimation.

\subsubsection{Proof techniques for $\frac12<q<1$}

The limitations of uniform reweighting motivate a different algorithm
for this regime. Our approach combines estimates of each probability
$p_i$ with quantum mean estimation. We outline its main components
and analysis below, beginning with a two-stage simultaneous probability
estimation procedure.

\paragraph{Two-stage simultaneous probability estimation.} Our probability estimation procedure builds on multidimensional
amplitude estimation~\cite{van2021quantum}.
A single global additive accuracy does not account for the
different precision requirements at different probability scales.
We therefore use two stages: the first obtains coarse estimates
and groups probabilities of comparable magnitude, while the
second refines the estimates within each group.
The resulting guarantee is summarized below.

\begin{theorem}[Two-stage simultaneous probability estimation, informal version
of Theorem~\ref{the: Two-stage multidimensional amplitude estimation}]
\label{thm:two-stage-informal}
Let $T\geq\sqrt n$.
Given purified quantum query access to a distribution $P$ on $[n]$,
there exists a quantum algorithm that uses
$\widetilde{O}(\sqrt n\,T)$ queries and outputs a classical vector
$r=(r_1,\ldots,r_n)$ such that, with probability at least $2/3$,
the following bounds hold simultaneously for every $i$:
\[
|r_i-p_i|=
\begin{cases}
O\!\left(\dfrac{\sqrt{p_i}}{T}\right),
& p_i\geq T^{-2},\\[6pt]
O(T^{-2}),
& p_i<T^{-2}.
\end{cases}
\]
Moreover,
\[
r_i=\Omega(T^{-2}),\qquad
p_i\leq2r_i,\qquad
\sum_{i=1}^n r_i=O(1),
\]
and
\[
\sum_{i=1}^n\frac{(p_i-r_i)^2}{r_i}
=O\!\left(\frac{n}{T^2}\right).
\]
\end{theorem}

The parameter $T$ controls both the query cost and the precision.
For probabilities above the cutoff $T^{-2}$, the additive error
scales as $\sqrt{p_i}/T$.
Below the cutoff, an additive error of order $T^{-2}$ suffices.
We also impose a floor of this order on the estimates, which
controls the negative powers of $r_i$ used in the entropy estimator.
The vector $r$ need not be a normalized probability distribution.

\paragraph{Estimating a correction under the input distribution.}
For $\frac12<q<1$, we use the probability estimates in a
first-order approximation, inspired by the approach
of \cite{chen2026nearly}.
Expanding each $p_i^q$ around $r_i$ gives
\[
p_i^q
\approx r_i^q+q r_i^{q-1}(p_i-r_i).
\]
Summing over $i$ yields
\[
F_q(P)\approx
(1-q)\sum_{i=1}^n r_i^q
+q\sum_{i=1}^n p_i r_i^{q-1}.
\]
The first term can be computed classically.
The second is an expectation under the original distribution:
\[
\sum_{i=1}^n p_i r_i^{q-1}
=\mathbb{E}_{I\sim P}[r_I^{q-1}],
\]
which we estimate using quantum mean estimation.

The guarantees of the probability estimation procedure control
both the approximation error and the cost of estimating this
expectation.
In particular, the approximation error satisfies
\[
0\leq
(1-q)\sum_{i=1}^n r_i^q
+q\sum_{i=1}^n p_i r_i^{q-1}
-F_q(P)
\leq O\left(\frac{n}{T^{2q}}\right).
\]
Choosing
\[
T=\Theta\left(
\left(\frac{n}{\varepsilon}\right)^{\frac{1}{2q}}
\right)
\]
makes this error sufficiently small.
The cost of constructing $r$ is therefore
\[
\widetilde{O}\left(
\frac{n^{\frac{1}{2}+\frac{1}{2q}}}{\varepsilon^{\frac{1}{2q}}}
\right).
\]

Furthermore, the bounds $p_i\leq2r_i$ and $\sum_i r_i=O(1)$
imply
\[
\operatorname{Var}_{I\sim P}(r_I^{q-1})
=O(n^{2-2q}).
\]
Quantum mean estimation consequently uses
$O(\frac{n^{1-q}}{\varepsilon})$ queries.
Combining the two costs gives
\[
\widetilde{O}\left(
\frac{n^{\frac{1}{2}+\frac{1}{2q}}}{\varepsilon^{\frac{1}{2q}}}
+\frac{n^{1-q}}{\varepsilon}
\right),
\]
as stated in Theorem~\ref{thm:high}.
Here, the conversion from power-sum error to Tsallis entropy
error is absorbed into the constants depending on $q$.

\subsection{Related work}

Quantum distribution testing includes algorithms for uniformity testing, closeness testing and testing similar means
\cite{chakraborty2010new,bravyi2011quantum,montanaro2015quantum,
gilyen2020distributional,luo2024succinct,canonne_et_al:LIPIcs.TQC.2025.7,gao2026quantum}.
These results establish quantum advantages for learning global
properties of distributions under suitable oracle access.
A related task is to estimate
numerical functionals of a distribution to a prescribed additive accuracy.

Under purified quantum query access, prior work on distributions includes estimators for R\'enyi entropy \cite{li2018quantum,wang2024quantum}, Shannon entropy \cite{gur2021sublinear,shin2025near}, and Tsallis entropy of integer orders \cite{wang2025information} and non-integer orders greater than two \cite{liu2026estimating}, as well as the multilevel framework of \cite{chen2026quantum}. Entropy estimation of quantum states has also been investigated in terms of query complexity \cite{subramanian2021quantum,wang2024new,chen2025list} and sample complexity \cite{acharya2020estimating,wang2025time,gao2026breaking,chen2026simultaneous,chen2025improved}. Our focus is on estimating the Tsallis entropy of classical distributions for $0<q<1$ under purified quantum query access. For a unified treatment, we refer readers to \cite{wang2026unified}.

\section{Preliminaries}
\subsection{Notations}

For a positive integer $n$, write $[n]=\{1,\ldots,n\}$.
Let $P=(p_1,\ldots,p_n)$ denote a probability distribution on $[n]$,
where $p_i\geq0$ and $\sum_{i=1}^n p_i=1$.
For a vector $x$, let $\|x\|_2$ and $\|x\|_\infty$ denote its
Euclidean and maximum norms, respectively. For a matrix $A$,
let $\|A\|$ denote its operator norm and $A^\dagger$ its adjoint.
We use $\widetilde{O}$ to suppress polylogarithmic factors in the
relevant problem parameters; a subscript $q$ indicates that the
implicit constants may depend on $q$.

\begin{definition}[Purified quantum access]
\label{def:access}
A probability distribution $P=(p_1,\ldots,p_n)$ is accessible
through purified quantum access if we are given a unitary $U_P$
satisfying
\begin{equation}
\label{eq:purified-access}
    U_P\ket{0}_G\ket{0}_I
    =
    \sum_{i=1}^n \sqrt{p_i}\ket{\phi_i}_G\ket{i}_I,
\end{equation}
where $G$ is an auxiliary register, $I$ is the outcome register,
and each $\ket{\phi_i}$ is a normalized quantum state.
\end{definition}

Throughout this paper, query complexity counts applications of
$U_P$, $U_P^\dagger$, and their controlled versions.

\begin{definition}[Hermitian dilation]
\label{def: Hermitian dilation}
Let $U$ be a unitary, and let $\Pi$ and $\widetilde{\Pi}$ be
orthogonal projectors such that
\[
    \widetilde{\Pi}U\Pi=A=
    \sum_{i=1}^{d}
    \sigma_i\ket{\widetilde{\psi}_i}\bra{\psi_i},
\]
where $d=\operatorname{rank}(A)$, $\sigma_i>0$, and
$\ket{\psi_i}$ and $\ket{\widetilde{\psi}_i}$ are the
corresponding right and left singular vectors.

Introduce an auxiliary qubit $h$ and define
\begin{align*}
    U_H
    &=
    \ket{0}\bra{1}_h\otimes U
    +
    \ket{1}\bra{0}_h\otimes U^\dagger,\\
    \Pi_H
    &=
    \ket{0}\bra{0}_h\otimes\widetilde{\Pi}
    +
    \ket{1}\bra{1}_h\otimes\Pi.
\end{align*}
Then $U_H$ is a Hermitian unitary, and its compression by
$\Pi_H$ is the Hermitian dilation of $A$:
\[
    \Pi_H U_H \Pi_H
    =
    \begin{bmatrix}
        0 & A\\
        A^\dagger & 0
    \end{bmatrix}.
\]
For each $i\in[d]$ and $\nu\in\{-1,+1\}$, define
\[
    \ket{\chi_{i,\nu}}
    =
    \frac{1}{\sqrt{2}}
    \left(
        \ket{0}_h\ket{\widetilde{\psi}_i}
        +
        \nu\ket{1}_h\ket{\psi_i}
    \right).
\]
These states satisfy
\[
    \Pi_H U_H \Pi_H\ket{\chi_{i,\nu}}
    =
    \nu\sigma_i\ket{\chi_{i,\nu}}.
\]
Consequently,
\[
    \Pi_H U_H \Pi_H
    =
    \sum_{i=1}^{d}\sum_{\nu\in\{-1,+1\}}
    \nu\sigma_i
    \ket{\chi_{i,\nu}}\bra{\chi_{i,\nu}},
\]
with any remaining eigenvectors having eigenvalue zero.
\end{definition}

\subsection{Quantum subroutines}

We recall the quantum subroutines used in our algorithms.
We begin with subset multidimensional amplitude estimation,
which allows us to estimate several probabilities simultaneously
and serves as the main building block of our two-stage procedure.

\begin{lemma}[Subset multidimensional amplitude estimation, Theorem 9 in \cite{van2021quantum}]\label{lem: Multidimensional Amplitude Estimation}
Let $P$ be a distribution on $[n]$ with the purified quantum access $U_P$. Let $S\subseteq[n]$ and $M\geq\sum_{i\in S}p_i$ with $0<M\leq 1$. Let $0<\epsilon, \delta<1$.
There exists a quantum algorithm $\textbf{SubsetEst}(U_P,S, M, \epsilon,\delta)$ that outputs a distribution $\tilde{P}$ such that $\Abs{\tilde{P}-P}_\infty\leq \epsilon$ with failure probability at most $\delta$ using
\begin{equation}\label{eq:subset}
 O\left(\frac{\sqrt M\ln{\frac{1}{\epsilon\delta}}}{\epsilon}\right)
\end{equation}
queries to $U_P$.
\end{lemma}

Once the probability estimates have been constructed, our
entropy estimator reduces the remaining computation to estimating
an expectation under the input distribution. For this step,
we use the following quantum mean estimation result.

\begin{lemma}[Quantum mean estimation, Theorem 1.1 in \cite{kothari2023mean}]\label{lem: mean estimation}
    Let $O_\Omega$ and $O_X$ be the quantum codes of a random variable $X$ on a finite probability space such that 
\begin{align*}
\begin{split}
    &O_\Omega\ket{\Vec{0}}=\sum_{\omega\in\Omega}\sqrt{p(\omega)}\ket{\omega}\ket{\psi_\omega},\\
    &O_X\ket{\omega}\ket{\Vec{0}}=\ket{\omega}\ket{x_\omega}\ket{\Vec{0}},
\end{split}
\end{align*}
where $\Omega$ is the domain of the probability distribution of $X$, $p:\Omega \to\sbra{0,1}$ 
is the probability distribution of $X$ and $\cbra{\ket{\psi_\omega}}$ are arbitrary normalized quantum states. Given $\epsilon,\delta\in\rbra{0,1}$, there exists a quantum algorithm that outputs an estimate $\tilde{\mu}$ such that $\Pr[|\tilde{\mu}-\mu|\geq \sigma\epsilon]\leq \delta$, where $\mu=\mathbb{E}\sbra{X}$ and $\sigma$ are the mean and standard deviation of $X$, respectively. The algorithm $O\rbra{\frac{\log\frac{1}{\delta}}{\epsilon}}$ applications of $O_X$ and $O_X^\dagger$.
\end{lemma}

We also recall the following amplitude estimation guarantee.

\begin{lemma}[Quantum amplitude estimation, Corollary 2.2 in \cite{chen2026quantum}]\label{lem: amplitude estimation}
    Let $A$ be a unitary and let $P$ be a projector such that $A|0\rangle = \sqrt{a}|\Psi_0\rangle + \sqrt{1-a}|\Psi_1\rangle$, where $a \in [0,1]$ and $P|\Psi_0\rangle=|\Psi_0\rangle$, $P|\Psi_1\rangle=0$. For $\epsilon,\delta\in\rbra{0,1}$, there exists a quantum algorithm that outputs $\hat{a}$ satisfying
    \begin{align}
        \Pr\left[|\hat{a} - a| \leq \epsilon\right] \geq 1-\delta.
    \end{align}
    The algorithm uses $O\rbra{\rbra{\frac{\sqrt{a}}{\epsilon}+\frac{1}{\sqrt{\epsilon}}}\log\frac{1}{\delta}}$ applications of $A$ and $P$.
\end{lemma}

The following lemma gives a procedure for distinguishing
singular values separated by a prescribed gap.
\begin{lemma}[Singular value discrimination, Theorem 3.1 in \cite{chen2026quantum}]\label{lem: Singular value discrimination}
Let $\widetilde{\Pi}U\Pi=A=\sum_{i=1}^{d}\sigma_i
\lvert\widetilde{\psi}_i\rangle\langle\psi_i\rvert$
be a singular value decomposition of $A$ with $\lVert A\rVert\leq 1/2$, and let
$\lvert\phi_{i,\nu}\rangle$ denote the corresponding eigenvector of the
Hermitian dilation of $A$, where $\nu\in\{-1,1\}$. For any $\varepsilon_1,\varepsilon_2\in(0,1]$ and any threshold
$\gamma\in(0,1/2)$, there exists a
unitary $D$ that, up to precision $\varepsilon_1$, implements
$D(\lvert ++\rangle_{ab}\lvert0\rangle_c\lvert\phi_{i,\nu}\rangle)_s
=\lvert ++\rangle_{ab}\lvert\xi_{i,\nu}\rangle_c
\lvert\phi_{i,\nu}\rangle_s$
for every $i\in[d]$ and $\nu\in\{-1,1\}$. Moreover,
$\lvert\xi_{i,\nu}\rangle$ satisfies
\[
\begin{cases}
\left\|
\lvert\xi_{i,\nu}\rangle-\nu\lvert0\rangle
\right\|\leq\varepsilon_2,
& \text{if }\sigma_i\in[\gamma,1),\\[4pt]
\left\|
\lvert\xi_{i,\nu}\rangle-\mathrm{i}\lvert1\rangle
\right\|\leq\varepsilon_2,
& \text{if }\displaystyle
\sigma_i\in\left[0,\frac{\gamma}{2}\right].
\end{cases}
\]
The unitary $D$ can be implemented using
$O(\log(\frac{1}{\varepsilon_1})+\frac{1}{\gamma}\log(\frac{1}{\varepsilon_2}))$
queries to controlled versions of $U$ and $U^\dagger$.
\end{lemma}

We next state a quantum singular value transformation result
for bounded polynomials of definite parity.
\begin{lemma}[Corollary~11 in \cite{gilyen2019quantum}]
\label{thm:qsvt}
Let $P\in\mathbb{R}[x]$ be a polynomial of degree $m$. Suppose that
$P$ has parity $m\bmod 2$, meaning that
$P(-x)=(-1)^mP(x)$, and that $|P(x)|\leq1$ for every
$x\in[-1,1]$. Then there exists a unitary
$U_{P}^{(\mathrm{SV})}$ such that
\[
P^{(\mathrm{SV})}\left(\widetilde{\Pi}U\Pi\right)
=
\begin{cases}
\left(\langle +\rvert\otimes\widetilde{\Pi}\right)
U_{P}^{(\mathrm{SV})}
\left(\lvert+\rangle\otimes\Pi\right),
& \text{if $m$ is odd},\\[5pt]
\left(\langle +\rvert\otimes\Pi\right)
U_{P}^{(\mathrm{SV})}
\left(\lvert+\rangle\otimes\Pi\right),
& \text{if $m$ is even}.
\end{cases}
\]
The unitary $U_{P}^{(\mathrm{SV})}$ can be implemented using $O(m)$
elementary gates, controlled applications of the reflections
$I-2\Pi$ and $I-2\widetilde{\Pi}$, and $O(m)$ queries to controlled
versions of $U$ and $U^\dagger$.
\end{lemma}

Finally, we recall a polynomial approximation result for
positive power functions.

\begin{lemma}[Polynomial approximation of positive power functions,
Lemma 6 in \cite{wang2024quantum}]
\label{lem:poly-approx-positive-powers}
Let $c>0$, $\beta\in(0,1)$, $\nu\in(0,\beta)$, and
$\eta\in(0,1/2)$. Define
$g(x)=2^{-c-1}\beta^{-c}x^c$. Then there exists an efficiently
computable even polynomial
$S=S_{c,\nu,\beta,\eta}\in\mathbb{R}[x]$ of degree
$O((c+1)\nu^{-1}\log(1/(\beta\nu\eta)))$ such that
\[
\begin{aligned}
|S(x)|&\leq 2g(x),
&&\text{for every }x\in[0,\nu],\\
|g(x)-S(x)|&\leq\eta,
&&\text{for every }x\in[\nu,\beta],\\
|S(x)|&\leq1,
&&\text{for every }x\in[-1,1].
\end{aligned}
\]
\end{lemma}

\section{An optimal estimator for $0<q\leq\frac{1}{2}$}\label{sec:uniform}
We adapt the multilevel framework of Chen et al.~(Section 3 of~\cite{chen2026quantum}) by changing the input state and the local QSVT polynomials. Uniform initialization assigns weight $1/n$ to each label, so the required transformation approximates a normalized positive power $x^q$. The discriminators and amplitude-estimation routine remain unchanged. We present the quantum algorithm in Algorithm~\ref{algo:uniform-multilevel} and prove its correctness and query complexity in the remainder of this section.

\begin{theorem}[Tsallis entropy estimation for $0<q\leq1/2$]
\label{thm:tsallis-low-complete}
Let $P=(p_1,\ldots,p_n)$ be a probability distribution over $[n]$ with the purified quantum access $U_P$. For any
$0<q\leq1/2$ and $0<\varepsilon\leq1$, there exists a quantum algorithm
that outputs an estimate $\widehat{H}$ satisfying $\left|
\widehat{H}
-
\mathrm{H}_q^{\text{Tsa}}(P)
\right|
\leq\varepsilon$ with probability at least $2/3$. The algorithm uses
$$
\widetilde{O}\left(
\left(\frac{n}{\varepsilon}\right)^{\frac{1}{2q}}
\right)
$$
queries to $U_P$ and $U_P^\dagger$.
\end{theorem}

\subsection{Uniform initialization and local transformations}
    
Set $\eta=(1-q)\varepsilon$. It suffices to estimate $F_q(P)=\sum_i p_i^q$ to additive error $\eta$. We first encode $\sigma_i=\sqrt{p_i}/2$ as singular values with explicitly preparable left singular vectors. Let $R|0\rangle=\frac12|0\rangle+\frac{\sqrt3}{2}|1\rangle$, and define
\begin{align}\label{eq:encoding}
\begin{split}
    U&=I_L\otimes U_P^\dagger\otimes R^\dagger,\\
\widetilde\Pi&=I_L\otimes|0\rangle\langle0|_G
 \otimes|0\rangle\langle0|_I\otimes|0\rangle\langle0|_r,\\
\Pi&=\sum_{i=1}^n|i\rangle\langle i|_L\otimes I_G
 \otimes|i\rangle\langle i|_I\otimes|0\rangle\langle0|_r.
\end{split}
\end{align}
Then
\[
\widetilde\Pi U\Pi=A
 =\sum_{i=1}^n\sigma_i|\ell_i\rangle\langle v_i|,
\qquad
|\ell_i\rangle=|i,0,0,0\rangle,
\quad |v_i\rangle=|i,\phi_i,i,0\rangle.
\]
Use the Hermitian dilation $(U_H,\Pi_H)$ from Definition~\ref{def: Hermitian dilation}, with the system register $s=(h,L,G,I,r)$. Its eigenvectors can be chosen as
\[
|\chi_{i,\nu}\rangle
=\frac{|0\rangle_h|\ell_i\rangle+\nu|1\rangle_h|v_i\rangle}{\sqrt2},
\qquad \nu\in\{-1,1\},
\]
with eigenvalues $\nu\sigma_i$. Let
\begin{align}\label{eq:input state}
    \begin{split}
        U_{\text{in}}\ket{0}=|\Psi_{\mathrm{in}}\rangle
&=|++\rangle_{ab}|00\rangle_{cd}\otimes
 \left(|0\rangle_h\otimes\frac1{\sqrt n}\sum_{i=1}^n|\ell_i\rangle\right)_s\\
&=|++\rangle_{ab}|00\rangle_{cd}\otimes
 \frac1{\sqrt{2n}}\sum_{i,\nu}|\chi_{i,\nu}\rangle_s.
    \end{split}
\end{align}
Observe that \(U_{\text{in}}\) makes no queries to \(U_P\). The discriminators use work qubits $a,b$ and flag qubits $c,d$; QSVT acts on $a$ and the system register $s$.
For the probability scales and error parameters, set
\[
t_j=4^{-j},\qquad C_q=2^{4q+2},\qquad
m=\min\{k\geq1:n(1+C_q)t_k^q\leq\eta/4\},
\qquad \xi=\frac{\eta}{1024mn(1+C_q)}.
\]
At level $j$, define
\[
I_j=(t_j/4,4t_j],\qquad
\Lambda_j=C_qt_j^q,\qquad
\alpha_j=\min\{1/4,\xi/(4\Lambda_j)\}.
\]
By Lemma~\ref{lem:poly-approx-positive-powers}, the even polynomial
\[
Q_j=S_{q,\,\sqrt{t_j}/4,\,\sqrt{t_j},\,\alpha_j}
\]
is bounded by one on $[-1,1]$ and approximates
\[
g_j(x)=2^{-q-1}t_j^{-q/2}x^q
\]
to error $\alpha_j$ on $[\sqrt{t_j}/4,\sqrt{t_j}]$. The normalization is chosen so that
\[
\Lambda_j g_j(\sqrt p/2)^2=p^q.
\]
Consequently, for $p\in I_j\cap[0,1]$,
\begin{equation}\label{eq:uniform-local-poly}
\left|\Lambda_jQ_j(\sqrt p/2)^2-p^q\right|
\leq2\Lambda_j\alpha_j\leq\xi.
\end{equation}
Thus each level approximates $p^q$ on its local interval; the discriminators select the appropriate levels.
\begin{algorithm}[htbp]
\caption{Uniform-input multilevel estimation of Tsallis entropy}
\label{algo:uniform-multilevel}
\begin{algorithmic}[1]
\REQUIRE Purified quantum access $U_P$ to $P=(p_i)_{i=1}^n$,
         $0<q\leq1/2$, and $0<\varepsilon\leq1$.
\ENSURE An estimate $\widehat H$ satisfying
        $|\widehat H-\mathrm{H}_q^{\mathrm{Tsa}}(P)|\leq\varepsilon$
        with probability at least $2/3$.

\STATE Set $\eta\gets(1-q)\varepsilon$, $C_q\gets2^{4q+2}$,
       and $t_j\gets4^{-j}$ for $j\geq0$.
\STATE Choose the smallest $m\geq1$ such that
       $n(1+C_q)t_m^q\leq\eta/4$.
\STATE Set $\xi\gets\eta/(1024mn(1+C_q))$ and $S\gets0$.
\STATE Construct $U_H$, $\Pi_H$, and $U_{\mathrm{in}}$ as described in  Eqs.~\eqref{eq:encoding} and~\eqref{eq:input state},
       where $U_{\mathrm{in}}\ket{0}=\ket{\Psi_{\mathrm{in}}}$.
\STATE Set $D_0\gets I$. For each $k\in[m]$, specify $D_k$
       for $(\Pi_H,\Pi_H,U_H)$ with threshold
       $\gamma_k=\sqrt{t_k}/2$, implementation precision $\xi/\sqrt{2n}$,
       and flag precision $\xi$.
\FOR{$j=1,\ldots,m$}
    \STATE Set $\Lambda_j\gets C_qt_j^q$,
           $\delta_j\gets\eta/(2mn\Lambda_j)$,
           and $\widehat a_j\gets0$.
    \IF{$\delta_j<1$}
        \STATE Set $\alpha_j\gets\min\{1/4,\xi/(4\Lambda_j)\}$
               and $Q_j\gets
               S_{q,\,\sqrt{t_j}/4,\,\sqrt{t_j},\,\alpha_j}$.
        \STATE Construct the QSVT unitary $W_j$ for $Q_j$
               on $(a,\mathsf{s})$, with operator-norm error
               at most $\xi$.
        \STATE Let $V_j=(W_j)_{as}
              (D_j)_{abds}\rbra{D_{j-1}}_{abcs}$, where $D_{j-1}$ is instantiated with flag qubit $c$ and $D_j$ with flag qubit $d$.
        \STATE Set the acceptance projector to
        \[
          M_j=
          \begin{cases}
          \ket{+}\bra{+}_a\otimes I_{bc}
          \otimes\ket{0}\bra{0}_d\otimes\Pi_H,
          &j=1,\\[3pt]
          \ket{+}\bra{+}_a\otimes I_b
          \otimes\ket{10}\bra{10}_{cd}\otimes\Pi_H,
          &j\geq2.
          \end{cases}
        \]
        \STATE Apply amplitude estimation
               to $V_j\ket{\Psi_{\mathrm{in}}}$ and $M_j$, obtaining
               an estimate $\widehat a_j$ of $a_j=\bigl\|M_jV_j\ket{\Psi_{\mathrm{in}}}\bigr\|^2$  with additive error $\delta_j$ and failure probability at most $1/(3m)$.
    \ENDIF
    \STATE $S\gets S+n\Lambda_j\widehat a_j$.
\ENDFOR

\RETURN $\widehat H=(S-1)/(1-q)$.
\end{algorithmic}
\end{algorithm}

\subsection{Correctness and query complexity}

\paragraph{Discriminator weights.} Ignoring implementation error, $D_j$ preserves $|\chi_{i,\nu}\rangle$ and outputs a normalized flag state $|z_{i,j,\nu}\rangle$. Let
\[
f_{i,j,\nu}=|\langle0|z_{i,j,\nu}\rangle|^2.
\]
The discriminator guarantee gives
\[
f_{i,j,\nu}\geq1-\xi^2\quad\text{if }p_i\geq t_j,
\qquad
f_{i,j,\nu}\leq\xi^2\quad\text{if }p_i\leq t_j/4.
\]
Set $f_{i,0,\nu}=0$ and define
\[
b_{i,j,\nu}=(1-f_{i,j-1,\nu})f_{i,j,\nu},
\qquad w_{i,j}=\frac12\sum_{\nu\in\{-1,1\}}b_{i,j,\nu}.
\]
Since $Q_j$ is even, the ideal acceptance probability is
\begin{equation}\label{eq:uniform-ideal-acceptance}
a_j^\star=\frac1n\sum_{i=1}^n w_{i,j}Q_j(\frac{\sqrt{p_i}}{2})^2.
\end{equation}
The weights satisfy two properties. First, they are localized:
\begin{equation}\label{eq:uniform-localization}
p_i\notin I_j\quad\Longrightarrow\quad w_{i,j}\leq\xi^2.
\end{equation}
Indeed, below $I_j$, $f_{i,j,\nu}$ is at most $\xi^2$; above $I_j$, $(1-f_{i,j-1,\nu})$ is at most $\xi^2$. For $j=1$, $p_i\notin I_1$ means that $p_i\leq \frac{1}{4}$, and thus $f_{i,j,\nu}$ is at most $\xi^2$.

Second, the relevant weights approximately sum to one. For $p_i>t_m$, choose the unique $k\in[m]$ with $t_k<p_i\leq t_{k-1}$, and let $K_i=\{k-1,k\}\cap[m]$. Then
\begin{equation}\label{eq:uniform-completeness}
1-\xi^2\leq\sum_{j\in K_i}w_{i,j}\leq1.
\end{equation}
For $k=1$, this is the guarantee for $f_{i,1,\nu}$. For $k\geq2$, suppress the indices $i,\nu$ and observe that
\[
b_{k-1}+b_k=(1-f_{k-2})f_{k-1}+(1-f_{k-1})f_k.
\]
Here $f_{k-2}\leq\xi^2$ (including $f_0=0$) and $f_k\geq1-\xi^2$, so the expression lies in $[1-\xi^2,1]$. Averaging over $\nu$ proves Eq.~\eqref{eq:uniform-completeness}. All other levels satisfy Eq.~\eqref{eq:uniform-localization}.

\paragraph{Approximation and implementation errors.} Define
\[
B_i=\sum_{j=1}^m\Lambda_jw_{i,j}Q_j(\sigma_i)^2,
\qquad n\sum_j\Lambda_ja_j^\star=\sum_iB_i.
\]
For $p_i>t_m$, Eqs.~\eqref{eq:uniform-local-poly}, \eqref{eq:uniform-localization}, and \eqref{eq:uniform-completeness}, together with $\Lambda_j\leq C_q$, imply
\[
|B_i-p_i^q|\leq\xi+\xi^2+mC_q\xi^2.
\]
For $p_i\leq t_m$, every level $j<m$ has weight at most $\xi^2$, while the last contributes at most $\Lambda_m$. Hence
\[
|B_i-p_i^q|\leq mC_q\xi^2+(1+C_q)t_m^q.
\]
Summing over labels yields
\[
\left|n\sum_j\Lambda_ja_j^\star-F_q(P)\right|
\leq n\xi+n(1+mC_q)\xi^2+n(1+C_q)t_m^q.
\]
Each discriminator contributes at most $\xi$ to the state error on the uniform superposition: its error on each of the $2n$ branches is at most $\xi/\sqrt{2n}$. Including QSVT, the total state error at a level is at most $3\xi$, so $|a_j-a_j^\star|\leq6\xi$. Therefore,
\[
\begin{aligned}
\left|n\sum_j\Lambda_ja_j-F_q(P)\right|
&\leq8nm(1+C_q)\xi+n(1+C_q)t_m^q\\
&\leq\frac{\eta}{128}+\frac{\eta}{4}<\frac{\eta}{2}.
\end{aligned}
\]

\paragraph{Estimation error.} By a union bound, all non-skipped estimates satisfy $|\widehat a_j-a_j|\leq\delta_j$ simultaneously with probability at least $2/3$. The same inequality holds deterministically at skipped levels because $a_j\in[0,1]$ and $\delta_j\geq1$. On this event,
\[
\left|S-n\sum_j\Lambda_ja_j\right|
\leq n\sum_j\Lambda_j\delta_j=\frac{\eta}{2}.
\]
Combining the two error bounds gives $|S-F_q(P)|\leq\eta$, and hence $|\widehat H-\mathrm H_q^{\mathrm{Tsa}}(P)|\leq\varepsilon$.

\paragraph{Query complexity.} The choice of the finest level gives
\[
t_m=\Theta\!\left(\left(\frac{\eta}{n}\right)^{\frac{1}{q}}\right),
\qquad m=O\!\left(\log\frac{n}{\eta}\right).
\]
At level $j$, the polynomial degree is
\[
\deg Q_j=O\!\left(\frac{q+1}{\sqrt{t_j}}
\log\frac{4}{t_j\alpha_j}\right)
=\widetilde O\!\left(t_j^{-\frac{1}{2}}\right).
\]
The discriminator at threshold $\gamma_j$ also costs
$\widetilde O(t_j^{-\frac{1}{2}})$ queries. Thus one application
of $V_j$ or $V_j^\dagger$ costs
$\widetilde O(t_j^{-\frac{1}{2}})$ queries.
Amplitude estimation to error $\delta_j<1$ uses
$O(\frac{\sqrt{a_j}}{\delta_j}+\frac{1}{\sqrt{\delta_j}})
=O(\frac{1}{\delta_j})$ applications, with an additional
$O(\log(3m))$ factor to reduce the failure probability to
$\frac{1}{3m}$. Consequently, the total query cost is
\[
\begin{aligned}
\widetilde O\!\left(\sum_{j:\delta_j<1}
\frac{t_j^{-\frac{1}{2}}}{\delta_j}\right)
&\leq\widetilde O\!\left(\frac{n}{\eta}
\sum_{j=1}^m t_j^{q-\frac{1}{2}}\right).
\end{aligned}
\]
For $q<\frac{1}{2}$, the geometric sum is
$O(t_m^{q-\frac{1}{2}})$, giving
\[
\widetilde O\!\left(\frac{n}{\eta}t_m^{q-\frac{1}{2}}\right)
=\widetilde O\!\left(
\left(\frac{n}{\eta}\right)^{\frac{1}{2q}}\right).
\]
For $q=\frac{1}{2}$, the sum equals $m$, so the cost is
$\widetilde O(\frac{n}{\eta})$. Since $\eta=(1-q)\varepsilon$,
both cases give the claimed bound.

\begin{remark}\label{rem:dualhigh}
For $1/2<q<1$, the same cost formula gives
\begin{equation}\label{eq:dualhigh}
 \widetilde O\!\left(\frac{n^{\frac32-q}}{\varepsilon}\right).
\end{equation}
\end{remark}

\begin{proof}
For $\frac12<q<1$, let $N_j=\min\{n,\frac{4}{t_j}\}$.
Localization gives $a_j\leq\frac{N_j}{n}+O(\xi)$, where
$\xi=O(\delta_j)$. Thus amplitude estimation yields a total cost
\[
\widetilde O\!\left(
\frac{\sqrt n}{\eta}\sum_{j=1}^m
\sqrt{N_j}\,t_j^{q-\frac12}
+\sqrt{\frac n\eta}\sum_{j=1}^m t_j^{\frac{q-1}{2}}
\right).
\]
Splitting at $t_j=\frac1n$ gives
\[
\sum_j\sqrt{N_j}\,t_j^{q-\frac12}
\leq
\sqrt n\sum_{t_j<\frac1n}t_j^{q-\frac12}
+2\sum_{t_j\geq\frac1n}t_j^{q-1}
=O(n^{1-q}).
\]
Also, $t_m=\Theta((\frac{\eta}{n})^{\frac1q})$ implies
\[
\sqrt{\frac n\eta}\sum_j t_j^{\frac{q-1}{2}}
=O\!\left(\left(\frac n\eta\right)^{\frac1{2q}}\right).
\]
Since $\left(\frac n\eta\right)^{\frac1{2q}}
\leq\frac{n^{\frac32-q}}{\eta}$, the cost is
\[
\widetilde O\!\left(\frac{n^{\frac32-q}}{\eta}\right)
=\widetilde O\!\left(\frac{n^{\frac32-q}}{\varepsilon}\right).
\]
\end{proof}

\section{Two-stage simultaneous probability estimation}\label{sec:pilot}

We now develop the two-stage simultaneous probability estimation outlined in
Theorem~\ref{the: Two-stage multidimensional amplitude estimation}.
The first stage uses coarse estimates to group probabilities by
magnitude, and the second refines the estimates within each group
at a precision determined by its scale. This separation allows us
to adapt the estimation accuracy to the individual probabilities
while controlling the total query cost.

\subsection{Main algorithm}

The following theorem collects the guarantees needed for the entropy
estimator in Section~\ref{sec:correction}. In addition to coordinatewise
accuracy, we require a positive lower bound on the estimates, control
of their total mass, and a bound on their weighted squared error.
These properties allow the estimates to serve as reference values
for approximating the power sum.

\begin{theorem}[Two-stage simultaneous probability estimation]\label{the: Two-stage multidimensional amplitude estimation}
Let $P$ be a distribution on $[n]$ with the purified quantum access $U_P$. Let $T\geq\sqrt n$. Then there exists a quantum algorithm that outputs a classical vector $r=(r_1, r_2 \cdots r_n)$ using $\tilde{O}(\sqrt n\,T)$ queries to $U_P$ such that the following properties hold:
\begin{enumerate}
    \item For every $i\in[n]$, $r_i\geq\frac{2}{T^2}$ and $p_i\leq2r_i$.
    \item For every $i\in[n]$, $|r_i-p_i|\leq \frac{\sqrt{p_i}}{8T}$ if $p_i> \frac{9}{4T^2}$, and $|r_i-p_i|\leq \frac{2}{T^2}$ if $p_i\leq\frac{9}{4T^2}$.
    \item $\sum_{i\in [n]}\frac{(p_i-r_i)^2}{r_i}\leq\frac{2n}{T^2}$.
    \item $\sum_{i\in [n]} r_i\leq4$. 
\end{enumerate}
\end{theorem}

Algorithm~\ref{alg:two-stage-mae} implements the two stages using
the subset-estimation subroutine from
Lemma~\ref{lem: Multidimensional Amplitude Estimation}.
In Stage I, $R_j$ contains the labels that remain unassigned at
level $j$, and $S_j$ contains those selected at scale $\epsilon_j$.
As the scale decays geometrically, progressively tighter upper bounds on the residual probability mass can be obtained, which in turn lowers the cost of subsequent calls. In Stage II, each selected group is estimated to accuracy $e_j$, proportional to
$\sqrt{\epsilon_j}/T$. Finally, a floor of $2\tau$ ensures that
every output coordinate is positive, including those corresponding
to labels left unassigned after Stage I.

\begin{algorithm}[hptb]
\caption{Two-stage Simultaneous Probability Estimation}
\label{alg:two-stage-mae}
\begin{algorithmic}[1]
    \STATE \textbf{Input:} Purified quantum access $U_P$ to a
    distribution $P$ on $[n]$, and $T\geq\sqrt n$.
    \STATE \textbf{Output:} A vector
    $r=(r_1,\ldots,r_n)$.
    \STATE Set $J\gets\log_2 T^2$, $\tau\gets \frac{1}{T^2}$,
    and $\delta_0\gets 1/(24J)$.
    \STATE \textbf{Stage I: probability grouping.}
    \STATE Initialize $R_0\gets[n]$ and $S_0\gets\varnothing$.
    \FOR{$j=0,\ldots,J-1$}
        \STATE Set $\epsilon_j\gets2^{-j}$.
        \IF{$R_j\neq\varnothing$}
            \STATE Set $M_j\gets\min\{1,2n\epsilon_j\}$.
            \STATE Obtain
            $\tilde{P}^{(j)}
            \gets\textbf{SubsetEst}
            (U_P,R_j,M_j,\epsilon_j/8,\delta_0)$.
            \STATE Set $S_j\gets\{i\in R_j:
            \tilde{p}_i^{(j)}\geq \frac{\epsilon_j}{2}\}$.
        \ENDIF
        \STATE Set $R_{j+1}\gets R_j\setminus S_j$.
    \ENDFOR

    \STATE \textbf{Stage II: higher-precision probability estimation.}
    \STATE Initialize $r_i\gets2\tau$ for every $i\in[n]$.
    \FOR{$j=0,\ldots,J-1$}
        \IF{$S_j\neq\varnothing$}
            \STATE Set $N_j\gets|S_j|$,
            $e_j\gets\frac{\sqrt{\epsilon_j}}{16T}$, and
            $M'_j\gets\min\{1,2N_j\epsilon_j\}$.
            \STATE Obtain $\hat{P}^{(j)} \gets\textbf{SubsetEst}(U_P,S_j,M'_j,e_j,\delta_0)$.
            \FOR{$i\in S_j$}
                \STATE Set $r_i\gets\max\{\widehat p_i^{(j)},2\tau\}$.
            \ENDFOR
        \ENDIF
    \ENDFOR
    \STATE Return $r$.
\end{algorithmic}
\end{algorithm}

\subsection{Proof of Theorem~\ref{the: Two-stage multidimensional amplitude estimation}}

\begin{proof}
Let $T\geq\sqrt n$, $T^2=2^J$, and $\tau=T^{-2}=2^{-J}$. We first bound the probability that all subset-estimation calls
succeed. On this event, we establish an invariant for Stage I
and use it to verify the four stated properties. We then bound
the query cost of each stage.

\paragraph{Success probability.}
Assign failure probability at most $1/(24J)$ to each
subset-estimation call. There are at most $J$ calls in each stage. Applying a union bound gives $\Pr[\text{at least one call fails}]
 \leq 2J\cdot\frac{1}{24J}
 =\frac{1}{12}.$
Thus all calls succeed simultaneously with probability at least
$11/12$. In the correctness analysis, we condition on this event.

\paragraph{Correctness.}
The key invariant is that every label remaining at level $j$
has probability at most twice the current scale. Specifically,
we prove by induction that, for every $i\in R_j$,
\begin{equation}\label{eq:invariant}
 p_i\leq2\epsilon_j.
\end{equation}
The base case $j=0$ follows from $p_i\leq1$ and $\epsilon_0=1$.
Suppose the invariant holds at level $j$. Then the probability
mass of $R_j$ satisfies
\[
 \sum_{i\in R_j}p_i
 \leq\min\{1,2|R_j|\epsilon_j\}
 \leq M_j.
\]
With Lemma~\ref{lem: Multidimensional Amplitude Estimation}, every selected label $i\in S_j$ satisfies
\begin{equation}\label{eq: p_i in S_j}
 \frac{3\epsilon_j}{8}
 \leq\tilde{p}_i^{(j)}-\frac{\epsilon_j}{8}
 \leq p_i\leq 2\epsilon_j.
\end{equation}
Conversely, every unselected label $i''\in R_{j+1}$ satisfies
\begin{align*}
     p_{i''}\leq \tilde{p}_{i''}^{(j)}+\frac{\epsilon_j}{8}
 <\frac{5\epsilon_j}{8}
 \leq2\epsilon_{j+1}.
\end{align*}
This completes the induction. The same argument at the final
level, where $\epsilon_{J-1}=2\tau$, shows that the remaining
labels satisfy $p_i<\frac{5\tau}{4}$ for every $i\in R_J$.

Having localized the probabilities in each selected group,
we turn to the refinement step. For each nonempty $S_j$, let
$N_j=|S_j|$. Equation~\eqref{eq: p_i in S_j} ensures that
$M'_j=\min\{1,2N_j\epsilon_j\}$ is a valid upper bound on its
probability mass. Applying
Lemma~\ref{lem: Multidimensional Amplitude Estimation} with
accuracy $e_j=\frac{\sqrt{\epsilon_j}}{16T}$ yields estimates that
satisfy $|\widehat p_i-p_i|\leq e_j$ for all $i\in S_j$.
Here and below, we suppress the superscript $(j)$ on the
refined estimates. The algorithm then sets
\begin{equation}\label{eq:floor}
 r_i=
 \begin{cases}
 \max\{\widehat p_i,2\tau\},
     &i\in\bigcup_{j=0}^{J-1}S_j,\\
 2\tau,&i\in R_J.
 \end{cases}
\end{equation}
We now verify the four properties in
Theorem~\ref{the: Two-stage multidimensional amplitude estimation}.
\paragraph{Claim 1: lower bounds on the estimates.} $r_i\geq\frac{2}{T^2}$ holds by construction. To prove $p_i\leq2r_i$, first consider $i\in S_j$. We have
\begin{align}
     \frac{e_j}{p_i}\leq\frac{\sqrt{\epsilon_j}}{16T}\frac{8}{3\epsilon_j}=\frac{1}{6T\sqrt{\epsilon_j}}
 \leq\frac{1}{6\sqrt2},
\end{align}
where the first inequality follows from Eq.~\eqref{eq: p_i in S_j}, and the last inequality follows from $\epsilon_j\geq \epsilon_{J-1}=2/T^2$. Consequently, $\hat{p}_i\geq p_i-e_j\geq p_i/2$, and hence
$p_i\leq2r_i$. For $i\in R_J$, this inequality follows from
 $p_i\leq\frac{5\tau}{4}$ and $r_i=2\tau$.

\paragraph{Claim 2: coordinatewise error.} For a label $i\in S_j$ with $\widehat p_i\geq2\tau$, $p_i\geq\frac{15}{8}\tau$. Therefore, the scale bound in
Eq.~\eqref{eq: p_i in S_j} relates the refinement accuracy
to the probability $p_i$ as follows:
\begin{align}
    \abs{p_i-r_i}=\abs{p_i-\hat{p}_i}\leq e_j=\frac{\sqrt{\epsilon_j}}{16T}\leq \frac{\sqrt{\frac{8p_i}{3}}}{16T}\leq \frac{\sqrt{p_i}}{8T}. 
\end{align}
For a label $i\in S_j$ with $\widehat p_i<2\tau$, $p_i\leq\frac{17}{8}\tau$. Thus we have 
\begin{align*}
    |p_i-r_i|=|p_i-2\tau|\leq 2\tau.
\end{align*}
For a remaining label $i\in R_J$, we have
$p_i<\frac{5\tau}{4}$ and $r_i=2\tau$, giving
\begin{align}
    \abs{p_i-r_i}\leq 2\tau=\frac{2}{T^2}.
\end{align}

\paragraph{Claim 3: weighted squared error.} We distinguish two cases according to whether the floor is active. If $i\in S_j$
and $\widehat p_i\geq2\tau$, then
$r_i=\widehat p_i\geq p_i/2$, and hence
\begin{align}
    \frac{(p_i-r_i)^2}{r_i} \leq\frac{2e_j^2}{p_i} \leq\frac{2\epsilon_j}{256T^2}\frac{8}{3\epsilon_j}=\frac{\tau}{48}.
\end{align}
If $i\in S_j$ and $\widehat p_i<2\tau$, then $r_i=2\tau$ and
$p_i\leq2\hat{p}_i<4\tau$. Therefore,
\begin{align}
    \frac{(p_i-r_i)^2}{r_i}=\frac{(p_i-2\tau)^2}{2\tau}\leq2\tau.
\end{align}
For $i\in R_J$, the same bound follows from
$0\leq p_i\leq5\tau/4$ and $r_i=2\tau$. Summing over all labels yields
\begin{align}
    \sum_{i=1}^n\frac{(p_i-r_i)^2}{r_i}\leq2n\tau.
\end{align}

\paragraph{Claim 4: total mass.}
For this calculation, write $e_i=e_j$ when $i\in S_j$ and
$e_i=0$ when $i\in R_J$. The refinement error and the floor
give $r_i\leq p_i+e_i+2\tau$, while
Eq.~\eqref{eq: p_i in S_j} gives
$e_i\leq\sqrt{p_i}/(8T)$. Summing over all labels and applying
the Cauchy--Schwarz inequality together with $T\geq\sqrt n$ yields
\begin{align}
\begin{split}
    \sum_{i=1}^n r_i
 &\leq1+\sum_{i=1}^n e_i+2n\tau\\
 &\leq1+\frac{1}{8T}
          \sum_{i=1}^n\sqrt{p_i}+\frac{2n}{T^2}\\
 &\leq1+\frac{\sqrt n}{8T}
          +\frac{2n}{T^2}
 \leq4.
\end{split}
\end{align}

\paragraph{Query complexity.}
It remains to bound the cost of constructing these estimates.
In Stage I, the mass bound $M_j$ decreases with the scale once
$2n\epsilon_j<1$. Splitting the sum at this threshold and
applying Lemma~\ref{lem: Multidimensional Amplitude Estimation},
we obtain a query count of
\begin{align}
\begin{split}
     &O\rbra{\sum_{j=0}^{J-1}
 \frac{\ln\frac{1}{\epsilon_j\delta_0}\sqrt{\min\{1,2n\epsilon_j\}}}{\epsilon_j}
}\\
&=O\rbra{\sum_{j: \epsilon_j\geq \frac{1}{2n}} \frac{\ln\frac{1}{\epsilon_j\delta_0}}{\epsilon_j}+\sum_{j: \epsilon_j<\frac{1}{2n}}\frac{\ln\frac{1}{\epsilon_j\delta_0}\sqrt{2n}}{\sqrt{\epsilon_j}}}\\
&=O\rbra{n\ln n+\sqrt{n}T\ln T}=\tilde{O}\rbra{\sqrt{n}T}.
\end{split}
\end{align}
For Stage II, the mass bound and the target accuracy both
depend on the group scale. Their ratio yields the following
cost for each nonempty group $S_j$:
\begin{align}
    O\rbra{\frac{\ln\frac{1}{e_j\delta_0}\sqrt{M'_j}}{e_j}}\leq \tilde{O}\rbra{\frac{\sqrt{2N_j\epsilon_j}}{\frac{\sqrt{\epsilon_j}}{16T}}}=\tilde{O}\rbra{T\sqrt{N_j}}.
\end{align}
The sets $S_j$ are disjoint, so $\sum_jN_j\leq n$. Applying the Cauchy--Schwarz inequality across the groups gives
\begin{align}
    \sum_{j=0}^{J-1}\sqrt{N_j}
 \leq\sqrt{J\sum_{j=0}^{J-1}N_j}
 \leq\sqrt{Jn},
\end{align}
and hence Stage II also uses $\tilde{O}(\sqrt n\,T)$ queries.
The notation $\tilde{O}$ absorbs factors depending
logarithmically on $T$ and the inverse per-call failure probability.
Combining the costs of the two stages gives the stated bound
of $\tilde{O}(\sqrt n\,T)$ queries.
\end{proof}

\section{Improved Tsallis entropy estimation for $\frac12<q<1$}\label{sec:correction}

We now apply the two-stage estimation procedure to estimate
the $q$-Tsallis entropy for $1/2<q<1$. Recall that
\[
 F_q(P)=\sum_{i=1}^n p_i^q,
 \qquad
 \mathrm{H}_q^{\mathrm{Tsa}}(P)
 =\frac{F_q(P)-1}{1-q}.
\]
Thus, an additive error of $\eta=(1-q)\varepsilon$ in the
power sum suffices for an additive error of $\varepsilon$
in the entropy.

Our estimator uses the classical vector $r$ from
Theorem~\ref{the: Two-stage multidimensional amplitude estimation} as the expansion point for a first-order
approximation to $F_q(P)$. The resulting expression consists
of a term computable from $r$ and an expectation under $P$,
which we estimate using quantum mean estimation. The weighted
squared-error guarantee for $r$ controls the approximation
remainder, while the bounds $p_i\leq2r_i$ and
$\sum_i r_i\leq4$ control the variance of the random variable used in mean estimation.
We analyze these two contributions separately before combining
them into the complete estimator.

\subsection{A first-order approximation with controlled remainder}

For $0<q<1$, concavity gives the tangent-line bound
\[
 x^q\leq y^q+qy^{q-1}(x-y)
 =(1-q)y^q+qxy^{q-1},
 \qquad x\geq0,\quad y>0.
\]
To apply this approximation, we require a quantitative bound on the remainder that remains valid at $x=0$. The following scalar inequality provides such a bound. It is inspired by the bias-control argument in Lemma~2.6 of \cite{chen2026nearly}. For completeness, we provide a direct proof.
\begin{lemma}[Concave power remainder]\label{lem:tangent}
For $0<q<1$, $x\geq0$, and $y>0$,
\begin{equation}\label{eq:tangent}
 0\leq(1-q)y^q+qxy^{q-1}-x^q
 \leq y^{q-2}(x-y)^2.
\end{equation}
\end{lemma}
\begin{proof}
Set $z=x/y$ and
$g(z)=(1-q)+qz-z^q$.
After dividing by $y^q$, it suffices to show that
$0\leq g(z)\leq(z-1)^2$ for $z\geq0$.
Nonnegativity follows from concavity. Since
$g(1)=g'(1)=0$ and
$g''(t)=q(1-q)t^{q-2}$ for $t>0$, Taylor's integral
formula gives, for $z\geq1$,
\[
 g(z)=q(1-q)\int_1^z(z-t)t^{q-2}\,dt
 \leq\tfrac12q(1-q)(z-1)^2.
\]
For $0<z<1$, the same formula gives
\[
 g(z)=q(1-q)\int_z^1(t-z)t^{q-2}\,dt.
\]
With the substitution $t=z+(1-z)u$, we have $t\geq u$.
Because $q-2<0$, it follows that
\begin{align*}
 g(z)
 &=q(1-q)(1-z)^2
   \int_0^1 u\bigl(z+(1-z)u\bigr)^{q-2}\,du\\
 &\leq q(1-q)(1-z)^2\int_0^1u^{q-1}\,du\\
 &=(1-q)(1-z)^2
 \leq(z-1)^2.
\end{align*}
The case $z=0$ follows by continuity. Multiplying by $y^q$
proves the claim.
\end{proof}

Using the estimates $\{r_i\}$ produced by the algorithm in Theorem~\ref{the: Two-stage multidimensional amplitude estimation}, define
\begin{equation}\label{eq:G}
 G(r,P)
 =(1-q)\sum_{i=1}^n r_i^q
 +q\sum_{i=1}^n p_i r_i^{q-1}.
\end{equation}
Lemma~\ref{lem:tangent} and Theorem~\ref{the: Two-stage multidimensional amplitude estimation} imply
\begin{align}
 0\leq G(r,P)-F_q(P)
 &\leq\sum_{i=1}^n r_i^{q-2}(p_i-r_i)^2\notag\\
 &=\sum_{i=1}^n r_i^{q-1}
       \frac{(p_i-r_i)^2}{r_i}\notag\\
 &\leq(2\tau)^{q-1}(2n\tau)
 =2^q n\tau^q
 \leq\frac{2n}{T^{2q}}.
 \label{eq:bias}
\end{align}
Here we used $q-1<0$ to bound
$r_i^{q-1}\leq(2\tau)^{q-1}$.
Thus, the floor on $r_i$ controls the growth of the derivative
near zero, and the weighted squared-error bound controls the
sum of the remainders. In particular, this argument also
covers coordinates for which $p_i=0$.

\subsection{Estimating the correction under the input distribution}

Write
\[
 A(r)=(1-q)\sum_{i=1}^n r_i^q.
\]
This term can be computed classically once $r$ is available.
For the remaining term in Eq.~\eqref{eq:G}, define the variable $W$ as 
\begin{equation}\label{eq:W}
  \Pr[W=r_i^{q-1}]=p_i,\qquad
 \mu=\mathbb E[W]=\sum_{i=1}^n p_i r_i^{q-1}.
\end{equation}
In particular, $G(r,P)=A(r)+q\mu$.

The cost of quantum mean estimation is governed by the
standard deviation of $W$. Using $p_i\leq2r_i$, we obtain
\begin{align}
 \mathbb E[W^2]
 &=\sum_{i=1}^n p_i r_i^{2q-2}\notag\\
 &\leq2\sum_{i=1}^n r_i^{2q-1}\notag\\
 &\leq2n^{2-2q}
       \left(\sum_{i=1}^n r_i\right)^{2q-1}\notag\\
 &\leq2\cdot4^{2q-1}n^{2-2q}.
 \label{eq:moment}
\end{align}
The third line follows from the concavity of $x^{2q-1}$,
since $2q-1\in(0,1)$. Consequently,
\begin{equation}\label{eq:tsallis-sigma}
 \sigma_W\leq \sqrt{\mathbb E[W^2]}
 \leq B_q n^{1-q},
 \qquad B_q=\sqrt{2\cdot4^{2q-1}}.
\end{equation}
This bound is independent of $T$: increasing the precision
of the reference vector reduces the approximation bias
without increasing the correction cost through this variance
bound. The condition $q>1/2$ enters precisely in the
concavity step in Eq.~\eqref{eq:moment}.

Since $r$ is known classically, the value $r_i^{q-1}$ can
be evaluated without further queries to $U_P$. We can
therefore apply Lemma~\ref{lem: mean estimation} to $W$,
using $U_P$ to prepare the distribution and
Eq.~\eqref{eq:tsallis-sigma} to bound its standard deviation.
This yields an estimate $\widehat\mu$ satisfying
\begin{equation}\label{eq:tsallis-mean-error}
 |\widehat\mu-\mu|\leq\frac{\eta}{2q}
\end{equation}
with conditional probability at least $11/12$ for every
fixed successful reference vector $r$. The number of queries is
\begin{equation}\label{eq:tsallis-mean-cost}
 O\!\left(\frac{q B_q n^{1-q}}{\eta}\right)
 =O\!\left(\frac{n^{1-q}}{\eta}\right).
\end{equation}

\subsection{Complete estimator and query complexity}

We now present the complete algorithm and prove its correctness and query complexity.

\begin{algorithm}[htbp]
\caption{Tsallis entropy estimation for $1/2<q<1$}
\label{alg:tsallis-high}
\begin{algorithmic}[1]
\STATE \textbf{Input:} Purified quantum query access $U_P$
to a distribution $P$ on $[n]$, $1/2<q<1$, and
$0<\varepsilon\leq1$.
\STATE Set $\eta\gets(1-q)\varepsilon$ and $T\gets\left(\frac{4n}{\eta}\right)^{1/(2q)}$.
\STATE Construct $r$ using Algorithm~\ref{alg:two-stage-mae} with $T$.
\STATE Compute $A(r)=(1-q)\sum_{i=1}^n r_i^q$ classically.
\STATE Apply Lemma~\ref{lem: mean estimation} to the variable $W$ defined in Eq.~\eqref{eq:W} to estimate $\mu$
to additive error $\eta/(2q)$ with failure probability
at most $1/12$, using the standard-deviation bound
$B_q n^{1-q}$.
Call the result $\widehat\mu$.
\STATE Set $\widehat F\gets A(r)+q\widehat\mu$.
\STATE Return $\widehat H=\frac{\widehat F-1}{1-q}$.
\end{algorithmic}
\end{algorithm}

\begin{theorem}[Tsallis entropy estimation for $1/2<q<1$]
\label{thm:tsallis-high-complete}
Let $P$ be a distribution on $[n]$ with purified quantum query access $U_P$ and $n\geq2$. For any $q\in(1/2,1)$ and $0<\varepsilon\leq1$, there exists a quantum algorithm that estimates
$\mathrm{H}_q^{\mathrm{Tsa}}(P)$ to additive error
$\varepsilon$ with probability at least $5/6$, using
\begin{equation}\label{eq:tsallis-high-complexity}
 \widetilde{O}\left(
\frac{n^{\frac12+\frac{1}{2q}}}
     {\varepsilon^{\frac{1}{2q}}}
+\frac{n^{1-q}}{\varepsilon}
\right)
\end{equation}
queries to $U_P$, $U_P^\dagger$, and their controlled
versions.
\end{theorem}

\begin{proof}
\textbf{Error bound.}
Condition on the success of both Algorithm~\ref{alg:two-stage-mae} and quantum mean
estimation, then
\[
 0\leq G(r,P)-F_q(P)
 \leq\frac{2n}{T^{2q}}\leq\frac{\eta}{2}.
\]
Equation~\eqref{eq:tsallis-mean-error} bounds the
contribution of mean estimation by
$q|\widehat\mu-\mu|\leq\eta/2$.
Combining the two error contributions gives
\begin{align*}
 |\widehat F-F_q(P)|\leq q|\widehat\mu-\mu|
       +|G(r,P)-F_q(P)|\leq\frac{\eta}{2}+\frac{\eta}{2}
 =\eta.
\end{align*}
Dividing by $1-q$ gives an entropy estimate with error at most $\varepsilon$.

\textbf{Success probability.}
The reference-vector construction fails with probability
at most $1/12$. For every successful output $r$, the
conditional failure probability of mean estimation is
at most $1/12$. Therefore,
\[
 \Pr[\text{either required event fails}]
 \leq\frac{1}{12}+\frac{1}{12}=\frac{1}{6}.
\]
This argument does not require independence between the
two stages.

\textbf{Query complexity.}
The reference-vector construction uses
$\widetilde O(\sqrt n\,T)$ queries. Hence, the query cost of constructing the reference vector is
\[
 \widetilde{O}\left(
\frac{n^{\frac12+\frac{1}{2q}}}
     {\eta^{\frac{1}{2q}}}
\right).
\]
By Eq.~\eqref{eq:tsallis-mean-cost}, estimating the
correction requires $O(\frac{n^{1-q}}{\eta})$ additional queries.
The vector $r$ is constructed once and reused throughout
mean estimation,
so these two costs add. Substituting
$\eta=(1-q)\varepsilon$ yields
\[
 \widetilde O\!\left(
 \frac{n^{\frac12+\frac{1}{2q}}}{((1-q)\varepsilon)^{\frac{1}{2q}}}
 +\frac{n^{1-q}}{(1-q)\varepsilon}
 \right),
\]
which gives Eq.~\eqref{eq:tsallis-high-complexity} for
fixed $q\in(1/2,1)$.
\end{proof}

\section{Discussion}

We have presented two approaches to estimating the
$q$-Tsallis entropy for $0<q<1$. For $0<q\leq1/2$, uniform
initialization within the multilevel framework of
\cite{chen2026quantum} yields the query complexity
$\widetilde{O}\left(
\left(\frac{n}{\varepsilon}\right)^{\frac{1}{2q}}
\right)$, which is optimal up
to logarithmic factors. For $1/2<q<1$, the two-stage simultaneous probability estimation procedure provides probability estimates whose
accuracy is adapted to the unknown magnitudes, and combining these
estimates with quantum mean estimation yields the complexity
$\widetilde{O}\left(
\frac{n^{\frac12+\frac{1}{2q}}}
     {\varepsilon^{\frac{1}{2q}}}
+\frac{n^{1-q}}{\varepsilon}
\right)$,
which is smaller than the bound of \cite{chen2026quantum} in the
high-precision regime $\varepsilon= O\rbra{n^{1-2q}}$.

Several questions remain open. For $1/2<q<1$, closing the $\sqrt n$
gap between our first term and the lower bound would require a
stronger reusable estimation primitive. It would also be of interest
to determine whether the two-stage simultaneous probability estimation procedure can be applied to other problems, such as
R\'enyi entropy estimation and $\ell_1$-distance estimation.

\section*{Declaration of AI Use}

The authors acknowledge the use of OpenAI's language models in
developing Theorem~\ref{the: Two-stage multidimensional amplitude estimation}, as well as in editing and improving the presentation
of the manuscript. All AI-assisted contributions were reviewed and
verified by the authors, who take full responsibility for the
correctness and integrity of the results presented in this paper.

\bibliographystyle{alpha}
\bibliography{refs}

\end{document}